\documentclass[11pt]{amsart}
\usepackage{amssymb, latexsym, amscd, amsmath}
\usepackage[mathscr]{eucal}
\usepackage{enumitem}
\usepackage[hmargin=0.75in,vmargin=0.75in]{geometry}

\usepackage[pdftex]{graphicx, color}

\usepackage[roman]{sublabel}

\usepackage[labelfont=bf]{caption}
\usepackage{subfig}

\definecolor{link_red}{rgb}{0.7,0,0}
\definecolor{cite_blue}{rgb}{0,0,0.97}

\usepackage{color}

 \usepackage{macros}

\theoremstyle{plain}
\newtheorem{theorem}{Theorem}[section]

\newtheorem{lemma}[theorem]{Lemma}

\newcommand{\reals}{\mathbb{R}}

\theoremstyle{definition}
\newtheorem{definition}[theorem]{Definition}

\newtheorem{remark}[theorem]{Remark}

\newtheorem{example}[theorem]{Example}

\begin{document}

\title{The Evolutionary stability of  partial migration with Allee effects}

\author{Yogesh Trivedi}
\address{Department of Mathematics\\
BITS Pilani, K.K Birla Goa Campus, 403726, Goa, India}
\urladdr[Anushaya Mohapatra]{}

\author{Ram Singh}
\address{Department of Mathematics\\
BITS Pilani, K.K Birla Goa Campus, 403726, Goa, India}
\urladdr[Anushaya Mohapatra]{}

\author{Anushaya Mohapatra}
\address{Department of Mathematics\\
BITS Pilani, K.K Birla Goa Campus, 403726, Goa, India}
\urladdr[Anushaya Mohapatra]{}


\subjclass[2000]{37}

\date{\today}

\begin{abstract}
An  Allee effect occurs when the per-capita growth rate increases at low densities.  Here, we investigate the evolutionary stability of a partial migration population with migrant population experiencing Allee effects.  Partial migration is a unique form of phenotypic diversity wherein migrant and non-migrant individuals coexist together. It is shown that when Allee effect is incorporated, the population undergoes a bifurcation as  the fraction of migrating population increases from zero to unity. Using an evolutionary game theoretic approach, we prove the existence of a unique evolutionary stable  strategy (ESS). It is also shown that the ESS  is the only ideal free distribution (IFD) that arises in the context of partially migrating population.

\end{abstract}

\maketitle

\section{Introduction}
Allee effect is a phenomenon, in which individual fitness increases with increasing density at low densities.  Many research has shown the  evidence of Allee effects and the importance of this biological phenomenon has been widely recognized in population dynamics, conservation programs, management of  endangered species and ecosystem dynamics \cite{boukal2002single,courchamp1999inverse,cushing2014backward,cushing2012evolutionary,edelstein2005mathematical,courchamp2008allee}. Here we study the evolution of partial migration phenomenon when only migrants experience Allee effects. Migration is a diverse phenomenon, and can be categorized into a multitude of forms. The most common type of migration is known as partial migration in which some individual migrate between habitats and others remain in a single habitat during their entire life.  Originally, these studies were motivated by birds, by now, partial migration has been found across many taxa, including fish, invertebrates and mammals \cite{chapman2011ecology,lande1976natural,lande1982quantitative,lundberg1988evolution,lundberg2013evolutionary}.  This type of within-population diversity is thought to play an important role in population stability and resilience \cite{schindler2010population}, so understanding how it is maintained by natural selection is critical for predicting how species may respond to future conditions.

There has been lot of interest in understanding such a complex system in which population consists of a mix of migratory and non-migratory individuals. Several mechanisms are studied in association with partial migration which  include genetic control, density-dependence, and exogenous stochastic effects in environmental variables \cite{chapman2011ecology,kokko2007modelling,lundberg1988evolution,ohms2019evolutionary,perez2014males}. 
In general, partial migration population models with negative density dependence effects are studied and neglecting the populations that experience Allee effects \cite{kokko2007modelling,de2017puzzle,mohapatra2016population,ohms2019evolutionary}. One of the most common examples of Allee effects occur when a  species is subject to predation with a saturating functional response, meaning that increased population levels decrease the risk of predation. Motivated by the partially migrating fish population steelhead rainbow trout system, in which steelhead experiences predation during its ocean phase, but trout does not, and only experiences the usual negative density dependence, we study a partial migration population model where the critical component is incorporation of an Allee effect. The primary goal of this work is to investigate what will happen when the biological mechanism Allee effects is included. Is partial migration preserved? Or can it be lost? How robust is it when the underlying population model is modified to incorporate the neglected biological feature?\\

We  investigate the evolutionary stability of the partial migration population with strong Allee effect (migrants) using evolutionary game theoretic or Darwin's dynamics approach. This is an alternative approach to Adaptive dynamics and is based  on Evolutionary Game Theory as given in \cite{lande1976natural,lande1982quantitative,allen2013adaptive,vincent2005evolutionary}.  
We also explore the connections between Ideal Free Distributions (IFD) arising from strategies that lead to partial migration behavior and evolutionary stable startegies.  It turns out  that the strategy which is the fraction of population became migrant act as a bifurcation parameter for the model. More ever it is shown that under rather general conditions the only possible strategies corresponding to an IFD, are  evolutionary stable, as well as convergent stable.

The rest of this paper is structured as follows. In Section 2 we formulate the population model and state the stability results. In Section 3 we derive an explicit formula for the ESS using evolutionary game approach. In Section 4 we discuss the results that connect IFDs and ESSs, and we conclude in Section 5 with possible future work and a brief non-mathematical discussion of the results obtained in this research.

\section{The population model with strong Allee effect}

 Consider the following stage structured, density-dependent and matrix population model.\\
\begin{equation}\label{sys-1}
\begin{pmatrix}
x_1(t+1)\\
x_M(t+1)\\
x_N(t+1)
\end{pmatrix}=
\begin{pmatrix}
0& f_M(z_M(t))& f_N(z_N(t))\\
\phi s_M& 0 & 0\\
(1-\phi)s_N& 0& 0
\end{pmatrix}
\begin{pmatrix}
x_1(t)\\
x_M(t)\\
x_N(t)
\end{pmatrix},
\end{equation}
where $x_1(t)$, $x_M(t)$ and $x_N(t)$ are non-negative real numbers, respectively representing the abundances of eggs, migrant adults and non-migrant adults at time $t$, where $t$ is a non-negative integer. A fraction $\phi\in[0,1]$ of eggs at time $t$ will become migrant adults, provided they survive a season, which is captured by the survival probability $s_M\in(0,1]$ in the model. Similarly, a fraction $1-\phi$ of eggs will become non-migrant adults, after surviving a season, with survival probability $s_N\in (0,1]$. The functions  $f_M(z)$ and $f_N(z)$ are per capita fertilities of migrants and non-migrants respectively. The parameter $\phi$ represents an allocation strategy whereby each morph (migrant or non-migrant) produces offspring that can become either type of morphs. 
$z_i$, with $i =M, N$ represents the total number of competing individuals experienced by phenotype $i$ during reproduction,
and it is given by 
\begin{equation} \label{zM}
z_M(t)= x_M(t) + p   x_N(t)
\end{equation} and 
\begin{equation} \label{zN}
z_N(t)= x_N(t) + q  x_M(t)
\end{equation} 
with $0 < p, q  < 1$. Here, $p$ is a parameter representing the fraction of the non-migrant population competing with each migrant adult. Similarly $q$ represents the fraction of the migrant population that competes with each non-migrant adult.

We now describe the  assumptions behind our general results:
\begin{enumerate}[leftmargin=*, topsep=6pt, itemsep=4pt, label={(\textbf{H\arabic*})}, ref=H\arabic*]

\item 
\label{Allee-type}

(\textbf{Migrant Allee effect}) $f_M :[0,\infty) \to [0,\infty)$  is a  smooth and unimodal function.  Namely there is an  unique positive population density $C_0$  such that  $f^{'}_M(C_0)=0$ and the maximum value of the function is $f_M(C_0)$. 
 $f_M$ increases at low densities i.e $f_M^{'}(z) >0$ for z sufficiently small. There is a positive equilibrium density A such that $f_M(z) < 1$ for  all $z <A$ and $f_M(z) > 1$ for some $z >A$. More ever we choose $s_Mf_M(C_0) > 1$ and  $g_M(z)=z f_M(z)$ has positive derivative.

\begin{remark}
In the case of $s_Mf_M(C_0) < 1$, the migrant only population model has only zero equilibrium point that may lead to extinction. So the assumption  $s_Mf_M(C_0) > 1$ guarantee the existence of non-zero or positive fixed points for migrant only population. The solution set of the equation $s_Mf_M(z)=1$  is non-empty, more ever it has precisely two element say $\hat x_M$ and $\tilde x_M$ with $ 0<\tilde x_M < C_0< \hat x_M$,  where $\hat x_M$ is referred as the carrying capacity of the migrant population.
\end{remark}
\begin{example}
	The most common Allee effect occurs in species subject to predation by a generalist predator with a saturating functional response.  Migrant population can be  modeled by the following equation
	
	\begin{equation}\label{$f_M(z)$}
	f_M(z)= \frac{az}{(b+z)^{2}} 
	\end{equation}
		
If we choose $a=5$ and $b=1$ then the graph of the function $f_M(z)$ is as shown in the following figure.

\begin{figure}[h]
	\begin{center}
\includegraphics[scale=0.70]{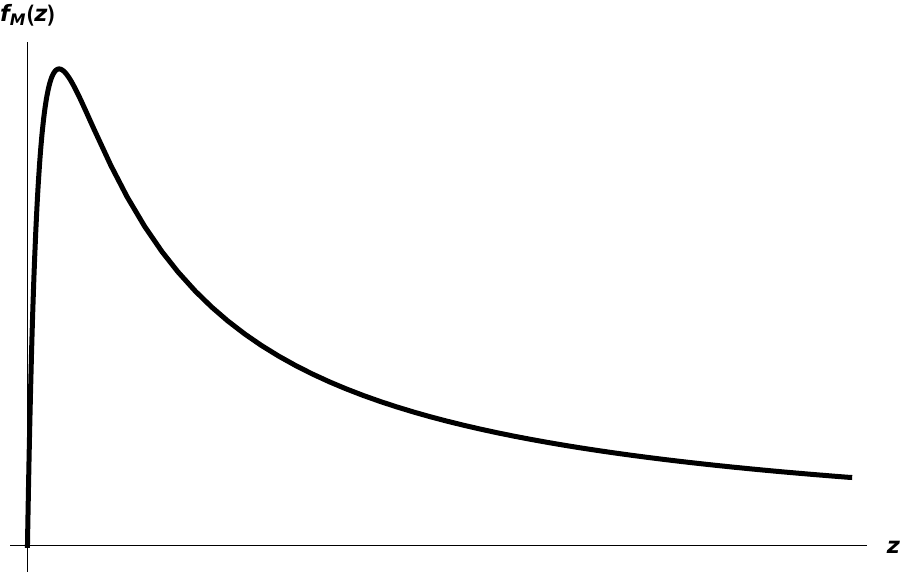}
\caption{The function $f_M$}
\label{fish}
\end{center}
	\end{figure}
\end{example}
\item
\label{monotone-sys}
 (\textbf{Decreasing fertility }) $f_N :[0,\infty) \to (0,\infty)$ is a smooth function, with negative derivatives and $g_N(z):=f_N(z)z$   has positive derivative. More ever $s_N f_N(0) >1$. Biologically, reproduction decrease with increasing population density.

\begin{remark}
	The condition  $s_N f_N(0) >1$ imply that when only non-migrant morph is present, the population persists. The equation $s_N f_N(z) =1$ has a unique solution $\tilde x_N > 0$ and referred as the carrying capacity of the non-migrant population.
	\end{remark}

\subsection{Stability Analysis}
\label{pd}

Model $(\ref{sys-1})$ can be re-written more compactly in vector form as 
\begin{equation}\label{sys-eq}
X(t+1)=A_1(X(t),\phi)X(t),
\end{equation}
where
$$
X=\begin{pmatrix}x_1\\x_M \\x_N\end{pmatrix},
A_1(X,\phi)=\begin{pmatrix}
0& f_M(z_M)& f_N(z_N)\\
\phi s_M& 0 & 0\\
(1-\phi)s_N& 0 & 0
\end{pmatrix},
$$ 
 $z_M(t)= x_M(t) + p   x_N(t)$ and $z_N(t)= x_N(t) + q   x_M(t).$
By splitting $A_1(X,\phi)$ as:
$$
A_1(X,\phi)=F+T,\textrm{ where } F=\begin{pmatrix}
0&f_M(z_M)&f_N(z_N)\\ 0 & 0 & 0\\0 & 0 & 0
\end{pmatrix},\textrm{ and } T=\begin{pmatrix}
0&0&0\\
\phi s_M& 0 & 0\\
(1-\phi) s_N &0&0
\end{pmatrix},
$$
allen2008basic,caswell2000matrix,cushing1998introduction,li2002applications
we can associate 
the basic reproduction number to the non-negative matrix $A_1(X,\phi)$ in the usual way \cite{allen2008basic,caswell2000matrix,cushing1998introduction,li2002applications}:
\begin{equation}
\label{R0:def}
R_0(X,\phi):=\rho(F(I-T)^{-1})=\phi R_0^M(z_M)+(1-\phi)  R_0^N(z_N)   \textrm{ for every } (X,\phi)\in \reals^3_+ \times [0,1].
\end{equation}
Here, $\rho(F(I-T)^{-1})$ denotes the spectral radius of $F(I-T)^{-1}$ and $ R_0^I(z_I) = s_I f_I(z_I)$ for $I= M, N.$\\

\begin{lemma}
	\label{origin-stability}
Suppose $\phi_1 = \frac{1-s_Nf_N(0)}{s_Mf_M(0)- s_Nf_N(0)}$, then for $\phi \in [0, \phi_1), R_0(0, \phi) > 1$, and for $\phi \in (\phi_1, 1], R_0(0, \phi)< 1$.
	\end{lemma}
\begin{proof}
We have $R_0(0, \phi)= \phi s_Mf_M(0) + (1-\phi)s_Nf_N(0)$. The derivative of $R_0(0, \phi)$ w.r.t $\phi$ is $s_Mf_M(0)- s_Nf_N(0)$ which is negative as $s_Mf_M(0) <1 < s_Nf_N(0)$ by hypothesis (\textbf{H1}) and  (\textbf{H2}). So $R_0(0, \phi)$ is a decreasing function of $\phi$.
More ever we have  $R_0(0, \phi) = 1 $ iff $\phi = \phi_1$,  where  $\phi_1 = \frac{1-s_Nf_N(0)}{s_Mf_M(0)- s_Nf_N(0)}.$ Hence for $\phi \in [0, \phi_1),  R_0(0, \phi) > 1 $  and  for  $\phi \in (\phi_1, 1],  R_0(0, \phi) < 1.$
\end{proof}
 \begin{lemma}
 	\label{L1}
 	Suppose the functions $f_M$ and $f_N$ satisfies the hypothesis  $({\bf \ref{Allee-type}})-({\bf H2})$ and suppose that $s_Nf_N(\frac{C_0}{p})>1, qs_M < s_N$ and  $ps_N < s_M$ . Define  $Q_{\phi}(x):=\phi s_M f_M(d_1x) + (1-\phi) s_N f_N(d_2x)$ where  $d_1(\phi) = \phi s_M + p (1-\phi) s_N$ and  $d_2(\phi) = q \phi s_M +  (1-\phi) s_N$.  Exactly one solution ${\bar x_1}(\phi)$ of the equation $Q_{\phi}(x)=1$  will lie in the interval $ (\frac{C_0}{d_1(\phi)}, \infty)$, where $f_M^{'}(C_0)=0$.
 \end{lemma}
 	\begin{proof}

Denote  $F_{\phi}^{M}(x) = 1 -\phi s_M f_M(d_1x)$ and $F_{\phi}^{N}(x) = (1 -\phi)s_N f_N(d_2x)$, then we have $$Q_{\phi}(x)-1 = F_{\phi}^{N}(x) - F_{\phi}^{M}(x)$$
Notice that $\frac{d}{dx}(F^M_{\phi})=0$ at $\frac{C_0}{d_1}$ and $F_{\phi}^{M}$ is a unimodal function for each $\phi \in (0, 1)$ as shown in the figure (\ref{R0g1}). Also $F_{\phi}^{N}$ is a strictly decreasing function in x.  So if $F_{\phi}^M( \frac{C_0}{d_1}) < F_{\phi}^N({\frac{C_0}{d_1}})$, then the  graphs of $F_{\phi}^M$ and $F_{\phi}^N$ will intersect at least at one point in the first quadrant. It is easy to check that the assumptions $s_N f_N(\frac{C_0}{p})>1, qs_M < s_N$, $ps_N < s_M$ and $s_Mf_M(C_0)>1 $ together implies  $F_{\phi}^M(\frac{C_0}{d_1}) < F_{\phi}^N ({\frac{C_0}{d_1}})$ . So the equation $Q_{\phi}(x) =1 $  has at least one solution in positive real line. Since $F^N_{\phi}$ is strictly decreasing and $F^M_{\phi}$ is strictly increasing on $[\frac{C_0}{d_1}, \infty)$, the equation $Q_{\phi}(x)=1 $ has exactly one solution $\bar x_1(\phi)$  in $(\frac{C_0}{d_1}, \infty)$ (see figure \ref{R0g1}.) 
\end{proof}

\begin{remark}
Lemma (\ref{L1}) says that the number of positive solution of the equation $Q_{\phi}(x) =1$  in $( \frac{C_0}{d_1(\phi)},\infty)$ is one, however there may be one or more then one positive solution of the equation $Q_{\phi}(x) =1$ in the interval $[0, \frac{C_0}{d_1})$. For simplification purpose we only focus on a class functions $f_N$, $f_M$ and other suitable parameters such that the equation  $Q_{\phi}(x) =1$ has at most one solution in $[0, \frac{C_0}{d_1})$ for all $\phi \in (0, 1)$ as shown in the figure (\ref{R0g1}) and also the main results in this work remain same even if  the equation $Q_{\phi}(x) =1$ has one or more solutions in  $[0, \frac{C_0}{d_1})$.

\end{remark}

\begin{figure} [h]

  \centering
  \subfloat[$\phi = \phi_1$]{\label{fig:fm}\includegraphics[width=0.40\textwidth]{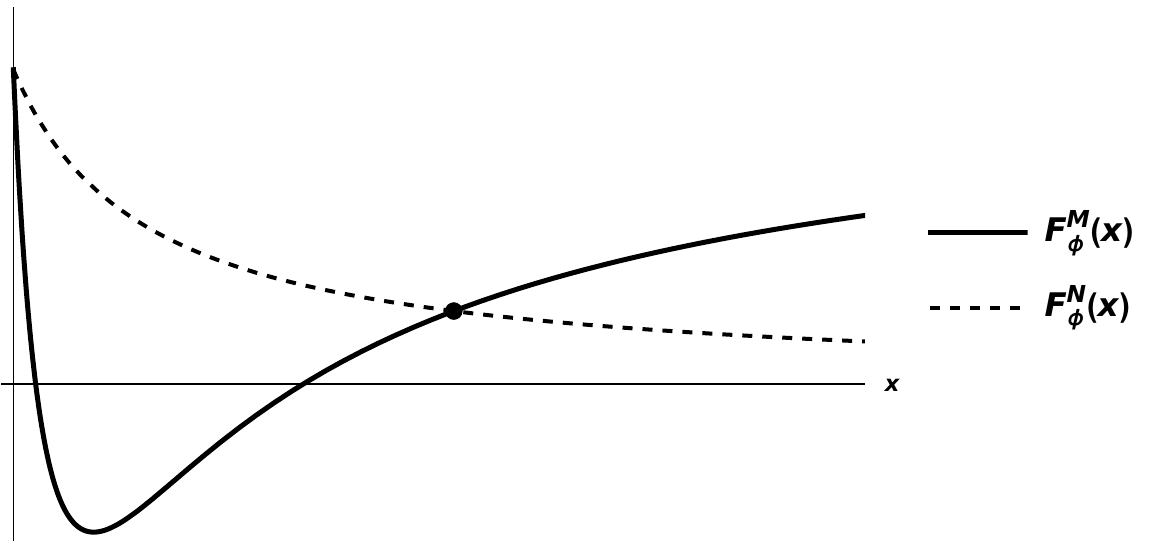}} \hspace{10pt}              
  \subfloat[$\phi \in (0, \phi_1)$]{\label{fig:gm}\includegraphics[width=0.35\textwidth]{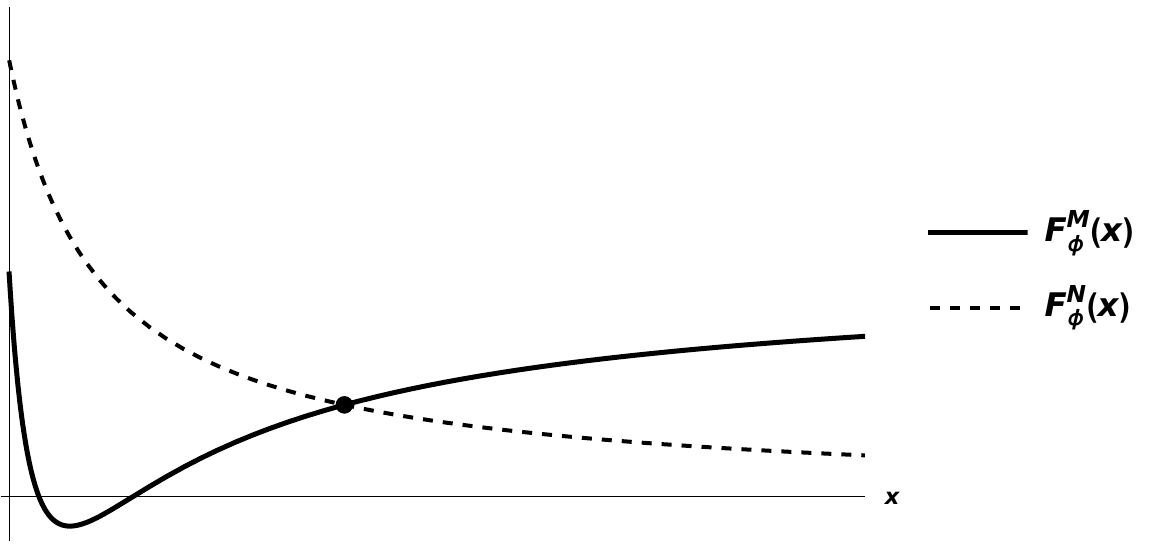}}
  \subfloat[$\phi \in (\phi_1, 1)$]{\label{fig:gm}\includegraphics[width=0.35\textwidth]{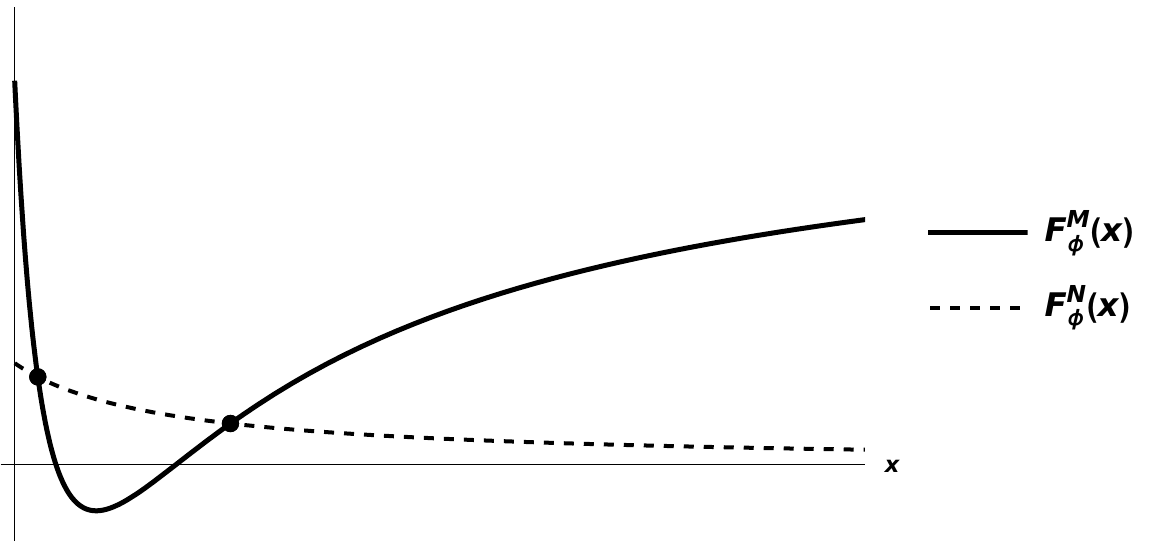}}
  	\label{r0a}
  \caption{} 
  \label{R0g1}
\end{figure}




\begin{remark}
The assumption $s_Nf_N(\frac{C_0}{p}) >1$  in lemma (\ref{L1}) is equivalent to $ C_0 < p \tilde x_N$, provided $qs_M < s_N$ and  $ps_N < s_M$, where $f_M^{'}(C_0) =0$.
\end{remark}

 So we make another hypothesis as follows.
\item 
\label{h3}
  For each $\phi \in (0, 1)$ the equation $Q_{\phi}(x) = 1 $ has at most one positive solution in the interval  $(0, \frac{C_0}{d_1})$ as shown in the figure (\ref{R0g1}),  where $Q_{\phi}(x)$ is the function as defined in lemma (\ref{L1}). And more ever as assumed in the lemma (\ref{L1})  choose $C_0 < p \tilde x_N, qs_M < s_N$ and  $ps_N < s_M$.

\end{enumerate}

\begin{example}
  Let $f_M(z)=\frac{7}{(1+z)^2}$ , $f_N(z)=\frac{3}{1+z}$ and
$ s_M=0.70, s_N=0.85, q=0.80,  p= 0.80$. In this case we get $\phi_1 =0.61$, so when  $\phi \in [0, 0.61)$ i.e. $R_0(0, \phi) > 1$,  the functions $F_{\phi}^M$ and $F_{\phi}^N$intersect at exactly one point. And if $\phi \in (0.61,1]$ i.e. $R_0(0, \phi) < 1$, the functions $F_{\phi}^M$ and $F_{\phi}^N$ intersect at two points. Hence the equation $Q_{\phi}(x)=1$ has one positive  solution for $R_0(0, \phi) > 1$, and two positive solutions for $R_0(0, \phi) < 1$.  The graphs corresponding to some $\phi$ values are given below. The functions $f_N$ , $f_M$ and other parameters satisfies all the hypothesis  (\textbf{\ref{Allee-type})} - (\textbf{\ref{h3}}).

\begin{figure} [h!]
	\centering
	
	\subfloat[$\phi =0.4$ i.e. 
	$R_0(0, \phi) > 1$]{\label{fig:fm}\includegraphics[width=0.40\textwidth]{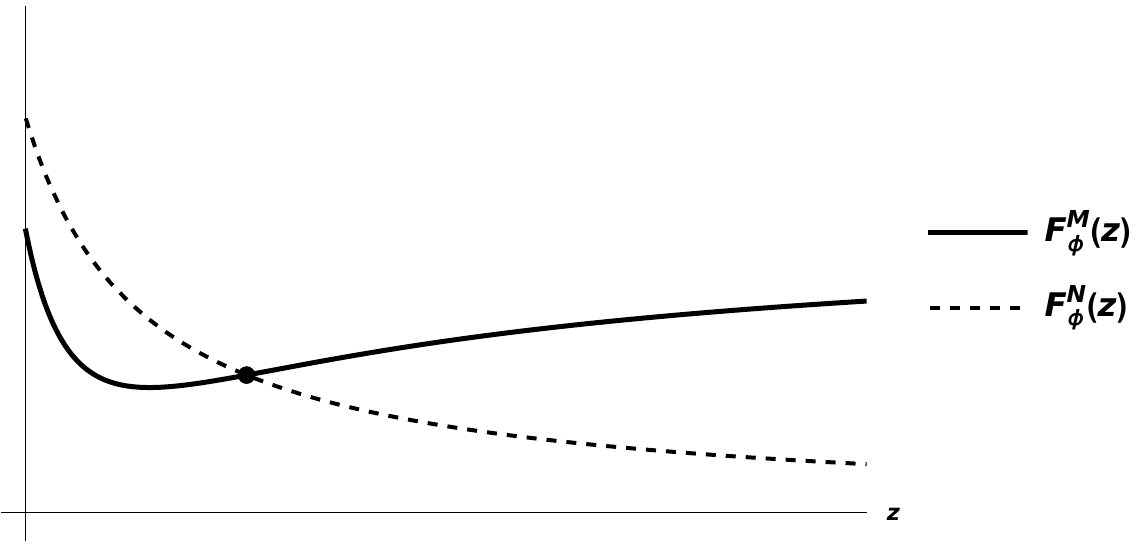}}             
	\subfloat[$\phi=0.75$ i.e. 
	$R_0(0, \phi) < 1$ ] {\label{fig:gm}\includegraphics[width=0.38\textwidth]{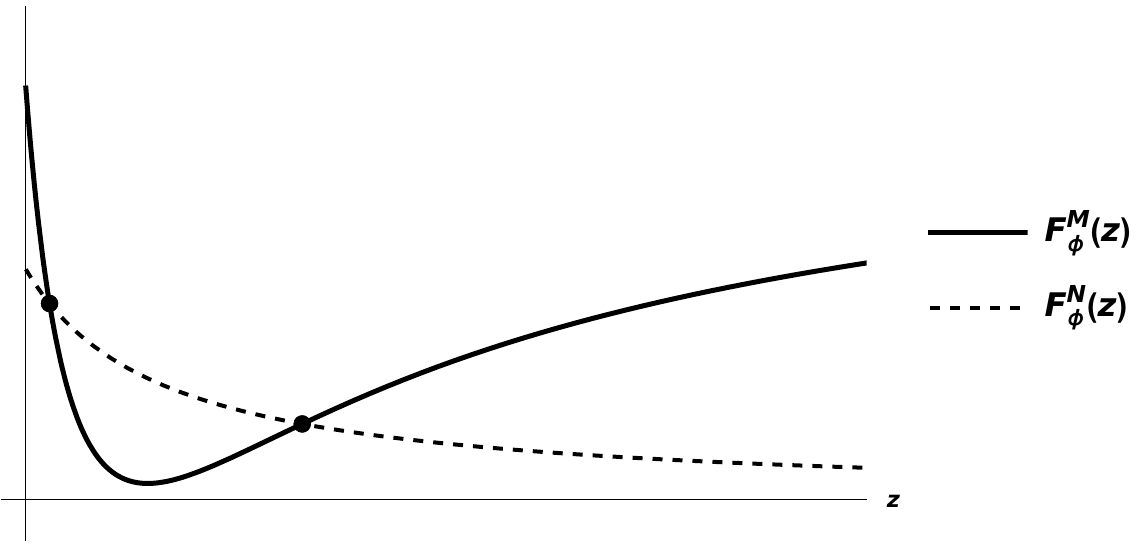}}

	\caption{} 
	\label{phi}
\end{figure}

\end{example}

\newpage

\begin{example}
Let $f_M(z)=\frac{8}{(1+z)^2}$ , $f_N(z)=\frac{9}{1+1.2z}$ and
$ s_M=0.60, s_N=0.28, q=0.45,  p= 0.90$. In this case we get $\phi_1 =0.60$, so when  $\phi \in [0, 0.60)$ i.e. $R_0(0, \phi) > 1$,  the functions $F_{\phi}^M$ and $F_{\phi}^N$ intersect at exactly one point. And if $\phi \in (0.60,1]$ i.e. $R_0(0, \phi) < 1$, the functions $F_{\phi}^M$ and $F_{\phi}^N$ intersect at two points. Hence the equation $Q_{\phi}(x)=1$ has one positive  solution for $R_0(0, \phi) > 1$, and two positive solutions for $R_0(0, \phi) < 1$. The graphs corresponding to some $\phi$ values are given below. The functions $f_N$ , $f_M$ and other parameters satisfies all the hypothesis  (\textbf{\ref{Allee-type})} - (\textbf{\ref{h3}})

\begin{figure} [!ht]
		\centering
		
		\subfloat[$\phi=0.50$ i.e. 
	$R_0(0, \phi) > 1$ ]{\label{fig:fm}\includegraphics[width=0.40\textwidth]{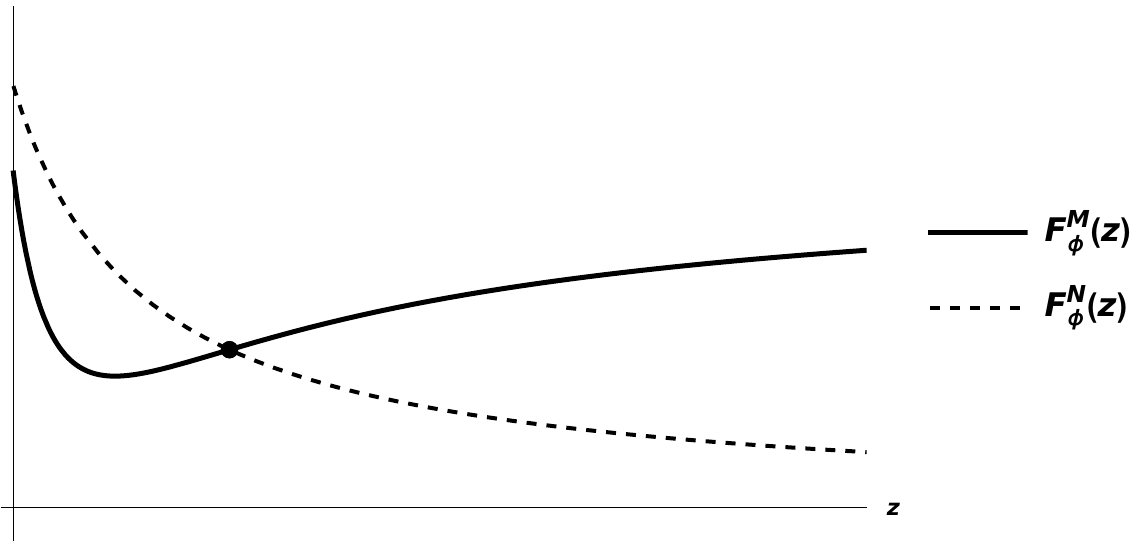}}               
		\subfloat[$\phi=0.85$ i.e. 
	$R_0(0, \phi) < 1$] {\label{fig:gm}\includegraphics[width=0.38\textwidth]{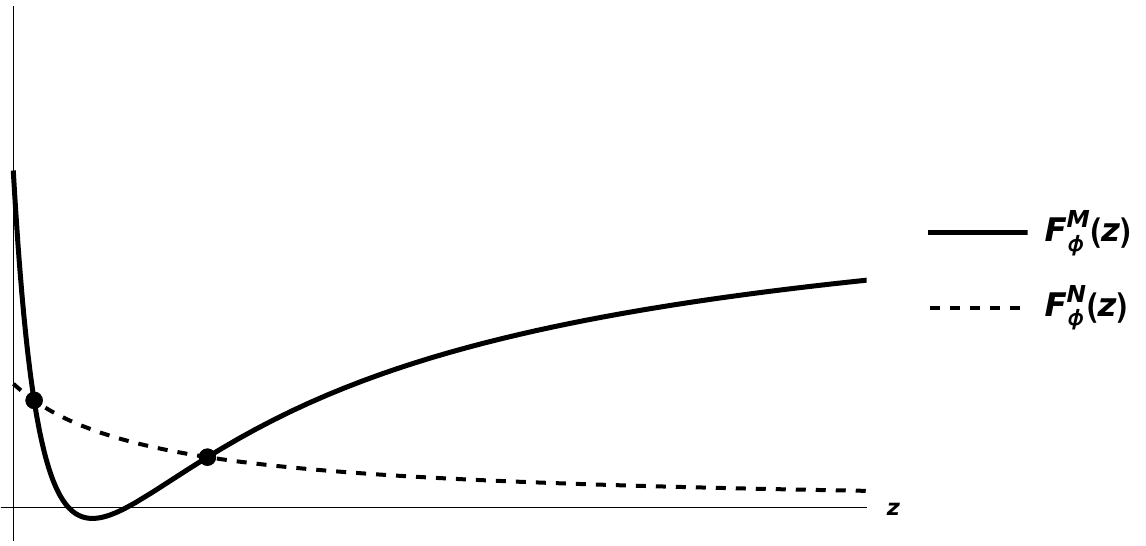}}
		
		\caption{} 
		\label{phi}
	\end{figure}
	
\end{example}

\subsection{The population dynamics}

 One can observe that if there exist a positive  fixed point $ x^* =(x_1^*, x_M^*, x_N^*)$  of the system \eqref{sys-eq}, then it  should satisfy the equation $R_{0}(\phi, x^{*})=1$. If we choose $d_1= \phi s_M + p (1-\phi) s_N $ and $d_2= q\phi s_M +  (1-\phi) s_N$, then the first coordinate $x_1^{*}$ satisfy the following
 \begin{equation}
 \label{Qx=1}
 Q_{\phi}(x)=\phi s_M f_M(d_1x) + (1-\phi) s_N f_N(d_2x) =1 
 \end{equation}
 Conversely suppose  any positive $x_1^*$ satisfy the equation \eqref{Qx=1}. If you choose $ x^* =(x_1^*, x_M^*, x_N^*) $ with $  x_M^*=\phi s_Mx_1^*$ and $ x_N^*=(1-\phi) s_Nx_1^*$, then  $R_0(x^*, \phi) =1 $ and $x^*$  is  a positive fixed point  of the model \eqref{sys-eq}.
 Since $x^*$ is uniquely defined by $x_1^*$, the  number of positive solution to the equation \eqref{Qx=1} is same as the number of positive fixed point of the system \eqref{sys-eq}. 

 We have the following result.
\begin{theorem} 
\label{pop}
Assume that the hypothesis (\textbf{\ref{Allee-type})} - (\textbf{\ref{h3}}) holds for the system \eqref{sys-1} then following holds with $\phi_1 = \frac{1-s_Nf_N(0)}{s_Mf_M(0)- s_Nf_N(0)}$:
\begin{enumerate}
\item For $\phi \in [0, \phi_1)$, the zero fixed point is unstable  and for $\phi \in (\phi_1, 1]$, it is locally  stable.
\vspace{0.1in}	
\item
If $\phi \in (0, 1)$, the system \eqref{sys-1} has exactly one locally stable positive fixed point $x^{*}(\phi)=(x_1^{*},  x_M^{*}, x_N^{*})$ with $x_1^{*} > \frac{C_0}{d_1}$, where $d_1 = \phi s_M + p(1-\phi) s_N$ and $f_M^{'}(C_0) = 0.$

\vspace{0.1in}
\item
 If $\phi=0$, then system $(\ref{sys-1})$ has a unique, non-zero locally stable fixed point $({\tilde x}_1, 0, {\tilde x}_N)$, where ${\tilde x}_N>0$ is the unique positive solution to the equation $s_N f_N(z)=1$, and If $\phi=1$, then system $(\ref{sys-1})$  has exactly one non-zero locally stable fixed point $({\hat x}_1, {\hat x}_M, 0)$, where ${\hat x}_M $ is a positive solution to the equation $s_M f_M(z)=1$, with $f_M^{'}({\hat x}_M) < 0$  and ${\hat x}_1={\hat x}_M/s_M$, 
${\tilde x}_1={\tilde x}_N/s_N$.

\end{enumerate}

\end{theorem}
\begin{proof}
	
For $\phi \in [0, \phi_1)$, $R_0(0, \phi) > 1$  and for  $\phi \in  (\phi_1, 1], R_0(0, \phi) < 1$ by Lemma (\ref{origin-stability}). So for $\phi \in  (\phi_1, 1]$ the origin is a locally stable fixed point and for $\phi \in [0, \phi_1)$, it is unstable.

 \textbf{Case I :} $\phi \in (0, 1).$ \\
 Suppose $R_0(0, \phi) \geq 1$, then by lemma (\ref{L1}) and hypothesis $(\textbf{H}3)$, for each $\phi$ the system \eqref{sys-1} has one positive fixed point $x^{*}(\phi)=(x_1^{*},  x_M^{*}, x_N^{*})$ with $x_1^{*} > \frac{C_0}{d_1}$.  The Jacobian matrix $J$ defined at the fixed point $(x_1^{*}, x_M^{*}, x_N^{*})$ is 
 \begin{equation}
 J=\begin{pmatrix}
 0 & F_1 & F_2 \\ \phi s_M & 0 & 0\\
 (1-\phi) s_N & 0 & 0
 \end{pmatrix}
 \end{equation} where $F_1=x_M^{*}{(t)}f'_M(z_M^{*}{(t)})+f_M(z_M^{*}{(t)})+q x_N^{*}{(t)}f'_N(z_N^{*}{(t)})$ and $F_2=px_M^{*}{(t)}f'_M(z_M^{*}{(t)})+ x_N^{*}{(t)}f'_N(z_N^{*}{(t)})+f_N(z_N^{*}{(t)})$. 
  The eigen values of the matrix J are found to be $0$, and  $\pm \sqrt{1 + x_1^{*}.Q_{\phi}'(x_1^{*})}$.
Where $$Q_{\phi}'(x_1^{*}) = d_1 \phi s_M f'_M(d_1x_1^{*}{(t)}) + d_2 (1-\phi) s_N f'_N(d_2x_1^{*}{(t)})$$
 with $d_1= \phi s_M + p (1-\phi) s_N $ and $d_2= q\phi s_M +  (1-\phi) s_N$.
Since the derivative of the functions $xf_M(d_1 x)$ and $xf_N(d_2 x)$ are positive, we have $\abs{d_1 x_1^{*} f'_M(d_1x_1^{*})} < \abs {f_M(d_1 x_1^{*})}$ and $\abs{d_2 x_1^{*} f'_N(d_2 x_1^{*})} < \abs {f_N(d_2 x_1^{*})}$. So $\abs{ x_1^{*}.Q_{\phi}'(x_1^{*})} <  \abs{d_1 \phi s_M x_1^{*} f'_M(d_1x_1^{*})} +  \abs{d_2 (1-\phi) s_N x_1^{*} f'_N(d_2x_1^{*})} < 1$. So the eigen value  $ \lambda = |\sqrt{1 + x_1^{*}.Q_{\phi}'(x_1^{*})}|$ of the matrix J is dominant, positive and less than one as $ x_1^{*} Q_{\phi}'(x_1^{*}) < 0$. Hence the fixed point $(x_1^{*}, x_M^{*}, x_N^{*})$ is locally asymptotically stable. Now suppose $R_0(0, \phi) < 1$, by lemma (\ref{L1}) and hypothesis $(\textbf{H}3)$, for each $\phi$ the system \eqref{sys-1} has two positive fixed points say $x^{**}(\phi)$ and $x^{*}(\phi)$.  Suppose $x_1^{**}$ and $x_1^{*}$ are the corresponding first coordinates with $x_1^{**} < \frac{C_0}{d_1} < x_1^{*}$.  Similar to the previous case (for $R_0(0, \phi) \geq 1$) the  eigen values of the Jacobian of the system \eqref{sys-1} at $x^{**}$ are $0,$ and $ \pm \sqrt{1 + x_1^{**}.Q_{\phi}'(x_1^{**})}$. In this case $x_1^{**}.Q_{\phi}'(x_1^{**})$ is positive as $Q_{\phi}(x) - 1$ is an increasing function at $x_1^{**}$, so the eigen value $\lambda = |\sqrt{1+ x_1^{**}.Q_{\phi}'(x_1^{**})}|$ is dominant, positive and bigger than one which make the fixed point $x^{**}(\phi)$ locally asymptotically  unstable . On the other hand $ x^{*}$ is locally stable using the same argument as in the case of $ R_0(0, \phi) \geq 1$ . So for each $\phi \in (0, 1)$, the system \eqref{sys-1} has exactly one locally stable fixed point $(x_1^{*}(\phi), x_M^{*}(\phi), x_N^{*}(\phi))$.\\

 \textbf{Case II :}
Suppose $\phi =0$, then notice that every orbit of $(\ref{sys-1})$ enters the invariant part of the the boundary of $\reals^3_+$ where $x_M=0$ in $1$ step. The restriction of the dynamics to this part of the boundary is given by a planar system:
\begin{equation}
\begin{pmatrix}
\label{planar}
x_1(t+1)\\
x_N(t+1)
\end{pmatrix}=\begin{pmatrix} 0&f_N(x_N(t))\\
s_N& 0\end{pmatrix} \begin{pmatrix}x_1(t)\\x_N(t) \end{pmatrix}
\end{equation}
Suppose that  $(x_1^{*}, x_N^{*})$ is a positive fixed point of the system \eqref{planar} then it turns out that   
 $s_Nf_N(x_N^{*}{(t)})= 1$
  which has a  unique positive solution say ${\tilde x}_N$. Choose  ${\tilde x}_1=\frac{{\tilde x}_N}{s_N}$, clearly $({\tilde x}_1, 0, {\tilde x}_N)$  the unique non-zero fixed point of the system \eqref{sys-1} when $\phi=0$. The Jacobian matrix when we 
 linearize the system \eqref{planar} near the fixed point $(\frac{{\tilde x}_N}{s_N},{\tilde x}_N)$  is given by
  \begin{equation}
      J_1=\begin{pmatrix}
0& {\tilde x}_N f_N{'}({\tilde x}_N)+f_N({\tilde x}_N)\\
 s_N& 0
\end{pmatrix}
  \end{equation}
  The  eigen values of the matrix $J_1$ are given by 
 $ \lambda= \pm \sqrt{1+s_N{\tilde x}_N f'_N({\tilde x}_N)}$.
Since the function $zf_N(z)$ is increasing and $f_N(z)$ is decreasing, $\mid  \lambda \mid< 1$, so the fixed point $({\tilde x}_1, 0, {\tilde x}_N)$ is linearly stable.

 Now suppose $\phi =1$, as similar to the case $\phi=0$, the fixed point is of the form $(x_1^{*}, x^{*}_M, 0)$, where $s_Mf_M(x^{*}_M) =1$ and $x_1^{*} = \frac{x_M^{*}}{s_M}$. The equation $s_Mf_M(x^{*}_M) =1$  has two solution ${\tilde x_M}$ and ${\hat x_M}$ satisfying ${\tilde x_M} < C_{0} < {\hat x_M}$ which corresponds to two fixed points ${\tilde{x}}$  and ${\hat {x}}$ of the system (9). The eigen values of the Jacobian, if we linearize the system near any fixed point $x^{*}=(x_1^{*}, x^{*}_M, 0)$ are $\lambda= \pm \sqrt{1+s_N x^{*}_M f'_M(x^{*}_M)}$. So $\abs{\lambda} > 1$ at $x^{*} = {\tilde x} $ as $f'_M({\tilde x_M}) >0 $, hence it is an unstable fixed point and $\abs{\lambda} < 1$ at $x^{*} = {\hat x_M} $ by the same argument as in the case $\phi=0$ so it is a stable fixed point.

  So in summary, we just proved that for each $\phi \in [0, 1]$, there is exactly one locally stable fixed point for the system \eqref{sys-1}. 
  
  \end{proof}

\section{Evolutionary game  and  ESS }
If Allee effect is present only in migrant population, the dynamics of the system (\ref{sys-1}) undergoes a bifurcation as the zero fixed point stability changes when parameter $\phi$ passes through $\phi_1$. The goal  is to investigate how the evolution of partial migration is affected when only the migrants experience Allee effects.  To address this question, we will use an evolutionary Game theory approach  as advocated in \cite{cushing1998introduction,lande1976natural,lande1982quantitative,vincent2005evolutionary}.
In this methodology, an individual's allocation strategy is denoted by $v$, and 
the {\it mean} allocation strategy $\phi(t)$ in the population is treated as a dynamic state variable whose dynamics are governed by Lande's equation (or the breeder's equation, Fisher's equation, or the canonical equation of evolution). The methodology provides a coupled system for the population dynamics and the mean allocation strategy, known as the Darwinian dynamics: 
\begin{eqnarray}
\begin{pmatrix}

x_1(t+1)\\
x_M(t+1)\\
x_N(t+1)
\end{pmatrix}&=&
\begin{pmatrix}
0& f_M(z_M(t))& f_N(z_N(t))\\
v s_M& 0& 0\\
(1-v)s_N& 0& 0
\end{pmatrix}\Biggr\rvert_{v=\phi(t)}
\begin{pmatrix}
x_1(t)\\
x_M(t)\\
x_N(t)
\end{pmatrix}\label{eg1}\\
\phi(t+1)&=&\phi(t)+\sigma^2\frac{\partial \ln(\lambda(x(t),v)}{\partial v}\Big\vert_{v=\phi(t)}\label{eg2},
\end{eqnarray}
Equation $(\ref{eg2})$ states that the change in the mean strategy is proportional to the fitness gradient. 
Fitness here is taken to be $\ln \left( \lambda(x,v) \right)$, where $\lambda(x,v)$ is the dominant eigenvalue of the matrix
$$
A_1(x,v)=\begin{pmatrix}
0&f_M(z_M)& f_N(z_N)\\
vs_M&0&0\\
(1-v)s_N&0&0
\end{pmatrix}.
$$
The constant $\sigma^2$ is related to the (assumed constant) variance of the strategy throughout the population (equal, or proportional to the variance, depending on how the trait dynamics are derived) and is referred to as the speed of evolution.

A straightforward calculation shows that  $\lambda(x,v)$ equals the square root of the basic reproduction number associated to $A_1(x,v)$, which we already defined in $(\ref{R0:def})$:
$$
\lambda(x,v)= (R_0(x,v))^{1/2},\textrm{ where } R_0(x,v):=v s_Mf_M(z_M)+(1-v)s_Nf_N(z_N).
$$
Hence, system $(\ref{eg1})-(\ref{eg2})$ can be re-written as:
\begin{eqnarray}
\label{sys-2}
\begin{pmatrix}
x_1(t+1)\\
x_M(t+1)\\
x_N(t+1)
\end{pmatrix}&=&
\begin{pmatrix}
0& f_M(z_M(t))& f_N(z_N(t))\\
\phi(t) s_M& 0& 0\\
(1-\phi(t))s_N& 0& 0
\end{pmatrix}
\begin{pmatrix}
x_1(t)\\
x_M(t)\\
x_N(t)
\end{pmatrix}\label{eq13}\\
\phi(t+1)&=&\phi(t)+\frac{1}{2}\frac{\sigma^2}{R_0(x(t),\phi(t))}\frac{\partial R_0(x(t),v)}{\partial v}\Big\vert_{v=\phi(t)}\label{eq14},
\end{eqnarray}
We first study this system for $\sigma=0$, i.e. when there are no evolutionary forces at work:
\begin{theorem}\label{sigma=0} Assume that $\sigma^2=0$. Suppose that 
	$({\bf H1})-({\bf H3})$ hold,
	
	Then the following holds:
	\begin{enumerate}
		\item For every fixed $\phi_0$ in $(0,1)$, there is a  positive fixed point $(x^*(\phi_0),\phi_0)$ such that 
		
		every positive solution of system $(\ref{eq13})-(\ref{eq14})$ with initial condition $(x_0,\phi_0)$ for $x_0$ in a neighborhood of $(x^*(\phi_0)$, converges to $(x^*(\phi_0),\phi_0)$. The fixed point $(x^*(\phi_0),\phi_0)$ 
		is locally stable with respect to positive initial conditions with arbitrary  positive $x_0$, but fixed $\phi_0$.
		\item
		If $\phi_0=0$, then every positive solution of system $(\ref{eq13})-(\ref{eq14})$ with initial condition $(x_0,\phi_0)$ for arbitrary positive $x_0$, converges to a unique non-zero fixed point $({\tilde x}_1, 0, {\tilde x}_N,0)$, where ${\tilde x}_N>0$ is the unique positive solution to the equation $s_Nf_N(z)=1$, and ${\tilde x}_1={\tilde x}_N/s_N$. This fixed point is linearly stable with respect to initial conditions with arbitrary  positive $x_0$, but fixed $\phi_0=0$.
		\item
		If $\phi_0=1$, then every positive solution of system $(\ref{eq13})-(\ref{eq14})$ with initial condition $(x_0,\phi_0)$ for arbitrary positive $x_0$, converges to a unique non-zero fixed point $({\hat x}_1, {\hat x}_M, 0,1)$, where ${\hat x}_M>0$ is a positive solution to the equation $s_Mf_M(z)=1$, and ${\hat x}_1={\hat x}_M/s_M$. This fixed point is linearly stable with respect to initial conditions with arbitrary positive $x_0$, but fixed $\phi_0=1$.
	\end{enumerate}
	
\end{theorem}
\begin{proof}
	The proof follows immediately from Theorem $\ref{pop}$, and the fact  that for each $\phi_0$ in $[0,1]$, the set 
	$\{(x,\phi)\in \reals^3_+ \times [0,1]\; | \; \phi=\phi_0\}$ is forward invariant for solutions of system $(\ref{eq13})-(\ref{eq14})$ when $\sigma^2=0$.
\end{proof}

\begin{definition}\label{ESS-def}
	Suppose $(x^{*}(\phi^{*}), \phi^{*})$ is the stable fixed point of the system $\eqref{eq13}- \eqref{eq14}$ for  $\phi^*$ in $[0,1]$ and for $\sigma = 0$. We say  $\phi^*$ is an evolutionary stable strategy (ESS) if $(x^*(\phi^*), \phi^{*})$ is a locally asymptotically stable fixed point of system $(\ref{eq13})-(\ref{eq14})$ for small positive $\sigma^2 \ne 0$. 
\end{definition}
This notion captures that if the population has adopted an ESS, then it can not be invaded by other population that use nearby strategies. Our next result shows the existence and uniqueness of an ESS.

\begin{theorem}\label{pos-sigma}
	
	Assume that $\sigma^2>0$. Suppose that $({\bf H1})-({\bf H3})$ hold, and the carrying capacities  satisfies the  inequalities :
	$  q {\hat x_M} < {\tilde x_N}$ and  $  p{\tilde x_N} <  {\hat x_M} $. Then the 
	system $(\ref{eq13})-(\ref{eq14})$ has a fixed point $(x^*(\phi^*),\phi^*)$ in $\reals^3_+ \times [0,1]$, where  
	 $x^*(\phi^*)$ is the unique stable positive fixed point of system $(\ref{sys-1})$ with $\phi=\phi^*$ (see Theorem $\ref{pop}$) and $\phi^*$ is given by the following formula
	 \begin{equation}
	 \label{ESS-value}
	 \phi^*=\frac{\frac{{\hat x}_M-p{\tilde x}_N}{{\tilde x}_N-q{\hat x}_M}}{\frac{{\hat x}_M-p{\tilde x}_N}{{\tilde x}_N-q{\hat x}_M}+\frac{s_M}{s_N}}.
	 \end{equation}
	  Moreover,  $(x^*(\phi^*),\phi^*)$ is a locally asymptotically stable fixed point of system $(\ref{eq13})-(\ref{eq14})$ for sufficiently  small positive $\sigma^{2}.$ So $\phi^{*}$ is a unique ESS.

\end{theorem}

\begin{proof}
	The system of equations, when evolutionary force at work are 
	\begin{eqnarray}
	\label{eq16}
	\begin{pmatrix}
	x_1(t+1)\\
	x_M(t+1)\\
	x_N(t+1)
	\end{pmatrix}&=&
	\begin{pmatrix}
	0& f_M(z_M(t))& f_N(z_N(t))\\
	\phi(t) s_M& 0& 0\\
	(1-\phi(t))s_N& 0& 0
	\end{pmatrix}
	\begin{pmatrix}
	x_1(t)\\
	x_M(t)\\
	x_N(t)
	\end{pmatrix}\\
	\phi(t+1)&=&\phi(t)+\frac{1}{2}\frac{\sigma^2}{R_0(x(t),\phi(t))}\frac{\partial R_0(x(t),v)}{\partial v}\Big\vert_{v=\phi(t)}\label{eq17},
	\end{eqnarray}
	
	$$\textrm{ where } R_0(x,v):=v s_Mf_M(z_M)+(1-v)s_Nf_N(z_N).
	$$
	$$\frac{\partial R_0(x(t),v)}{\partial v}\Big\vert_{v=\phi(t)}\label{eg}=s_Mf_M(z_M)-s_Nf_N(z_N)$$

	Suppose $(x_1^{*}, x_M^{*}, x_N^{*}, \phi^{*})$ is a fixed point for the system \eqref{eq16}-\eqref{eq17}. We have

	$$\phi^{*}(t)=\phi^{*}(t)+\frac{1}{2}\frac{\sigma^2 (s_Mf_M(z^{*}_M)-s_Nf_N(z^{*}_N))}{\phi^{*} s_Mf_M(z^{*}_M)+(1-\phi^{*})s_Nf_N(z^{*}_N)}$$
	which implies that
	$s_Mf_M(z^{*}_M)-s_Nf_N(z^{*}_N)=0$ so 
	$s_Mf_M(z^{*}_M)=s_Nf_N(z^{*}_N)$.
	We also have $$R_0(x^{*}, \phi^{*}) =\phi^{*}(s_Mf_M(z^{*}_M)-s_Nf_N(z^{*}_N))+s_Nf_N(z^{*}_N)=1$$
It follows that  $s_Nf_N(z^{*}_N)=1$ and $s_Mf_M(z^{*}_M)=1$.
 The equation $s_Nf_N(z_N)=1$ has unique solution say  ${\tilde x}_N$ and  $s_Mf_M(z_M)=1$ has two solution ${\tilde x}_M$ and ${\hat x}_M$ with  $f'({\tilde x}_M)>0 $, $f'({\hat x}_M)<0$ and ${\tilde x}_M < {\hat x}_M.$
 Hence $\text{either} \;\:  z^*_M(\phi)= x^*_M(\phi) + p x^*_N(\phi)={\hat x}_M  \;\:\text{or}\;\:   z^*_M(\phi) = {\hat x}_M \textrm{ and }z^*_N(\phi)=x^*_N(\phi) + q x^*_M(\phi)={\tilde x}_N$
	After simplifying and solving for $x_M^*$ and $x_N^*$ in case of $z^*_M(\phi) = {\hat x}_M$  we get
	\begin{eqnarray}
	x^*_M(\phi) &=& \frac{{\hat x_M} - p {\tilde x_N}}{1-pq} \label{formula1} \\
	x^*_N(\phi) &=&\frac{{\tilde x_N} - q {\hat x_M}}{1-pq}  \label{formula2}
	\end{eqnarray}
	it follows that:
	\begin{equation}
	\label{ratio}
	\frac{x^*_M(\phi)}{x^*_N(\phi)}=\frac{{\hat x_M} - p {\tilde x_N}}{{\tilde x_N} - q {\hat x_M}}=\frac{s_M}{s_N}\frac{\phi}{1-\phi}.
	\end{equation}
	
	\begin{equation}  \label{IFD-value}
	\phi^*= \frac{\frac{{\hat x}_M-p{\tilde x}_N}{{\tilde x}_N-q{\hat x}_M}}{\frac{{\hat x}_M-p{\tilde x}_N}{{\tilde x}_N-q{\hat x}_M}+\frac{s_M}{s_N}}.
	\end{equation}
	Similarly in case of $z^*_M(\phi) = {\tilde x}_M$, we have 
	
	\begin{eqnarray}
	x^*_M(\phi) &=& \frac{{\tilde x_M} - p {\tilde x_N}}{1-pq} \label{formula1} \\
	x^*_N(\phi) &=&\frac{{\tilde x_N} - q {\tilde x_M}}{1-pq}  \label{formula2}
	\end{eqnarray}
	
Because $\tilde x_M < C_0< p\tilde x_N$, the corresponding fixed point is not a positive one. So we have one positive fixed point for the system \eqref{eq16}-\eqref{eq17}, we will show that the fixed point $(x^*(\phi^*),\phi^*)$  is  local asymptotically  stable. We linearize the Darwinian system (16)-(17) near fixed point $(x^*(\phi^*),\phi^*)$ yielding the following Jacobian matrix:\\
	$J_D(\sigma^2)=$
	$\begin{pmatrix}
		0& a& b& 0\\
		c& 0& 0& g\\
		d& 0& 0& -h\\
		0& e& f& 1
	\end{pmatrix}$
	where $a=f_M'(z_M^*)x_M^*+f_M(z_M^*)+x_N^*f'_N(z_N^*).q$, $ b= pf_M'(z_M^*)x_M^*+f_N'(z_N^*)x_N^*+f_N(z_N^*)$, 
		$c=	\phi^*s_M$, $ g=s_Mx_1^*$ , $d=	(1-\phi^*)s_N$, $h=s_Nx_1^*$, $e=\frac{\sigma^2}{2}(s_Mf'_M(z_M^*)-qs_Nf'_N(z_N^*))$ and $f=\frac{\sigma^2}{2}(ps_Mf'_M(z_M^*)-s_Nf'_N(z_N^*)).$
	
In the appendix, we have shown that the dominant eigenvalue of the matrix $J_D(\sigma^2)$ is positive and less than one, hence the $\phi^{*}$ given by the equation $\eqref{IFD-value}$ is an unique ESS.

\end{proof}

\section{The ESS is an IFD  in two habitat environment}
Theorem $\ref{pos-sigma}$ revels the existence of an evolutionary stable strategy $\phi^{*} \in (0, 1)$  such that for any population with the fixed strategy $\phi^{*}$, can not be invaded by other population with near by strategies. However, it does not say whether that strategy which lead to stable fixed points that correspond to an ideal free distribution. The IFD has emerged in many studies on the evolution of dispersal \cite{cantrell2020evolution,cantrell2008approximating,cressman2004ideal,gadgil1971dispersal,hastings1983biological,kvrivan2008ideal,mcpeek1992evolution}. Here we study the IFD concept for a  population that exhibits partial migration.
The goal of this Section is to examine this issue. The concept of an IFD is based on the ability of individuals to assess the quality of a spatial  environment, yet model \eqref{sys-1} does not incorporate space explicitly. However, it does contain space implicitly, and we can re-interpret model \eqref{sys-1} as a population model that evolves in two habitats $h_M$ and $h_N$, representing a migratory and non-migratory habitat. To each habitat we can associate a pay-off or fitness function which depends on the density of the abundances in that habitat, at equilibrium state. A population at equilibrium is considered to be ideal free, if the finesses in both habitats are equal and maximal. In this scenario, individuals would have no incentive to move to a different habitat. In this section, we will show that  the ESS is in fact an IFD, corresponding to a non-extreme value of the strategy parameter $\phi^{*}$. But first, we give a mathematical definition of the IFD concept.
\begin{definition}
Recall from Theorem $\ref{pop}$ that for each $\phi \in [0, 1]$, there is  one positive equilibrium state $(x^*_1(\phi), x_M^*(\phi), x_N^*(\phi))$ which is locally asymptotically stable for system \eqref{sys-1}. We define the \textbf{payoff (or fitness)} function for habitat $h_I$ to be $R_0^{I}(z^{*}(\phi^{*}))=s_I f_I (z_I^{*})$, where $I=M, N$ , $z^*_M(\phi)= x^*_M(\phi) + p x^*_N(\phi)$ and  $z^*_N(\phi)= x^*_N(\phi) + q x^*_M(\phi)$. 
 An allocation strategy  $\phi^*$  is said to be an \textbf{IFD} strategy for model $\eqref{sys-1}$ if the two habitats $h_M$ and $h_N$  have the same payoffs  at $\phi^{*}$,  i.e if 
$$R_0^M(z^*_M(\phi^*)) =R_0^N(z^*_N(\phi^*)),$$ 
and if no other strategy satisfying this payoff equality condition, has a higher payoff.
\end{definition}

\begin{theorem}
\label{ifd-thm}
If the conditions of Theorem \ref{pos-sigma} hold for system (\ref{sys-1}) 
 then the ESS  strategy $\phi^*$ given by the formula (\ref{ESS-value}) is the only IFD.
\end{theorem}

\begin{proof}
	By formula (\ref{ESS-value}), we have the following;
 $$\frac{{\hat x_M} - p {\tilde x_N}}{{\tilde x_N} - q {\hat x_M}}=\frac{s_M}{s_N}\frac{\phi^*}{1-\phi^*}.$$
 
 Choose {\Large$\tilde x_1=  \frac{{\hat x_M} - p {\tilde x_N}}{(1-pq)\phi s_M} = \frac{{\tilde x_N} - q {\hat x_M}}{(1-pq) (1-\phi) s_N}$}
 and $ {\tilde x}= ({\tilde x_1}, \frac{{\hat x_M} - p {\tilde x_N}}{1-pq} , \frac{{\tilde x_N} - q {\hat x_M}}{1-pq})$. 
It is easy to check that 
  $R_0({\tilde x}, \phi^*) = 1$ so, clearly {\Large$ {\tilde x}$ } is an equilibrium point of the system (\ref{sys-1}) which is positive. One can show that it is in fact locally asymptotically stable. So we have  

$$
x^*_M(\phi^*)= \frac{{\hat x_M} - p {\tilde x_N}}{1-pq} \textrm{ and }x^*_N(\phi^*)= \frac{{\tilde x_N} - q {\hat x_M}}{1-pq}
$$
Hence the pay-off functions are found to be equal:
$$R_0^M(z^*_M(\phi^*))= R_0^M({\hat x_M})=1=R_0^N({\tilde x_N})=R_0^N(z^*_N(\phi^*))$$ and this proves that  $\phi^*$ is indeed an IFD.

Conversely suppose $\phi=\phi^{*}$ is an IFD. Then $$R_0^M(z^*_M(\phi^*))=R_0^N(z^*_N(\phi^*))$$
By the theorem (\ref{pos-sigma}), $\phi^{*}$ is an ESS.
 \end{proof}
\section{discussion}

 The goals of this study were twofold: investigating if the evolution of partial migration is affected when only the migrants experience Allee effects and using the evolutionary game theory approach to determine evolutionary stable strategies if any.   In existing literature  \cite{griswold2011equilibrium,de2017puzzle,lundberg2013evolutionary,ohms2019evolutionary}, it is shown that negative density dependence in the fertilities alone can explain the partial migration phenomenon, provided it is attenuated with increasing sub-type abundances. The works are mostly numerical. Nevertheless, several biological features specially Allee effects of partially migrating populations have been neglected. An Allee effect is a positive association between absolute average individual fitness and population size over some finite interval.  In some partially migrating population, the migrant individual  experience predation during their stays at migrant habitat which may result in Allee effects. When the size of populations subject to  Allee effects is low, then these populations tend towards extinction. This fact argues for a thorough understanding of Allee effects and their mechanisms in order to develop sound management practices for a number of environmental issues. So the important aspect of this work is that, we have obtained an analytical formula for the Evolutionary Stable Strategy (ESS) for the allocation strategy adopted by a partially migrating population with migrant population experiencing a strong Allee effects. This ESS is expressed in terms of the demographic model parameters for the migrant and non-migrant populations, and thus the formula can be used to predict the ESS value, whenever the life histories of the migrant and non-migrant populations are known, for example from lab or empirical data.  More ever, the conditions on environmental parameters like carrying capacities under which the population will be evolutionary stable are clearly stated. Our results differ from previous studies  \cite{griswold2011equilibrium,de2017puzzle,lundberg2013evolutionary,ohms2019evolutionary} in that partial migration as it point to more of a system
of thresholds, which provides some insight into how populations may respond to future conditions when migrant undergoes Allee effects. We also investigate the connection between spatio/temporal structure and the evolution of partial migration behavior by associating the ESS to Ideal Free distribution (IFD) which happens to be a powerful tool for understanding  how populations distribute themselves in heterogeneous environments.  Mostly the phenomenon is studied in temporally constant environment but many environments are seasonal and changes with time. So for the future work it is motivating  to investigate the phenomenon when the underlying population model is modified to incorporate time varying environments.
\section{Acknowledgment}
This work is funded by start up research grant, SERB (SRG/2019/002200) to AM. 

\section{Appendix}

\textbf{The dominant eigen value of the Jacobian of Darwinian system:}\\

The Jacobian matrix with $\sigma^2 > 0$ is given by
$$
J_D(\sigma^2)=\begin{pmatrix}
C(x^*(\phi^*),\phi^*)&{\bf *}\\
0&1
\end{pmatrix}+
\sigma^2
\begin{pmatrix}
0&0&0&0\\
0&0&0&0\\
0&0&0&0\\
0&\frac {(s_Mf'_M(z_M^*)-qs_Nf'_N(z_N^*))}{2}&\frac{(ps_Mf'_M(z_M^*)-s_Nf'_N(z_N^*))}{2}&0
\end{pmatrix},
$$
where the ${\bf *}$ represents a $3$-dimensional vector whose value is unimportant at present, and the $3\times 3$  matrix $C(x^*(\phi^*),\phi^*)$ is defined as:
$$
C(x^*(\phi^*),\phi^*)=\begin{pmatrix}
0 & a & b\\
c&0&0\\
d&0&0
\end{pmatrix}, 
$$
with $a=f_M'(z_M^*)x_M^*+f_M(z_M^*)+x_N^*f'_N(z_N^*).q, b= pf_M'(z_M^*)x_M^*+f_N(z_M^*)+x_N^*f'_N(z_N^*), c=\phi^{*}s_M $ and $d= (1-\phi^*)s_N$. One can check that the eigen values of the matrix $C(x^*(\phi^*),\phi^*)$ are $0$ and $\pm \sqrt{1+ x_1^{*} Q^{'}(x_1^{*})}$. We have proved in the theorem (2.10) that the dominant eigen value  $ |\sqrt{1+ x_1^{*} Q^{'}(x_1^{*})}|$ of  the matrix $C(x^*(\phi^*),\phi^*)$ is positive and less than one.
Consequently, the dominant eigenvalue of $J_D(0)$ is $1$, and this eigenvalue is simple. By continuity of eigenvalues, the matrix $J_D(\sigma^2)$ will also have a real, simple and dominant eigenvalue $\lambda_p(\sigma^2)$ for all sufficiently small $\sigma^2$, such that $\lambda_p(0)=1$. We claim that $\lambda_p(\sigma^2)<1$, at least for all sufficiently small $\sigma^2$. To prove this, we now examine the roots of the characteristic polynomial $F(\lambda,\sigma^2):=\det (\lambda I -J_D(\sigma^2))$ associated to the matrix $J_D(\sigma^2)$. 
\begin{equation}\label{constants}
J_D(\sigma^2)=\begin{pmatrix}
0&a&b&0\\
c&0&0&g\\
d&0&0&-h\\
0&\sigma^2 e& \sigma^2 f& 1
\end{pmatrix}
\end{equation}
Where  $ g=s_Mx_1^*$ ,  $ h=s_Nx_1^*$, $e=\frac{\sigma^2}{2}(s_Mf'_M(z_M^*)-qs_Nf'_N(z_N^*))$ and $f=\frac{\sigma^2}{2}(ps_Mf'_M(z_M^*)-s_Nf'_N(z_N^*)).$
A tedious calculation shows that the characteristic polynomial of the matrix $J_D(\sigma^2)$ is given by:
$$
F(\lambda,\sigma^2)=\lambda^4-\lambda^3-\left[\sigma^2(eg-fh)+(ac+bd) \right]\lambda^2+(ac+bd)\lambda+\sigma^2(be-af)(ch+dg)
$$
Note that $F(\lambda,0)$ is positive for all $\lambda>1=\lambda_p(0)$ (since $\lambda_p(0)=1$ is the {\it dominant} root of $F(\lambda,0)$, and $\lim_{\lambda\to \infty}F(\lambda,0)=+\infty$). Moreover, $\partial F/\partial \lambda (\lambda_p(0),0)$ must be positive because $\lambda_p(0)=1$ is a {\it simple} root of $F(\lambda,0)$ (this can also be shown directly by calculating this partial derivative using the expression above: 
$\partial F/\partial \lambda (\lambda_p(0),0)=1-(ac+bd)$, and this is positive). 
Therefore, the claim above (namely, that $\lambda_p(\sigma^2)<1$, for all sufficiently small $\sigma^2$) will be proved, provided we can show that for all sufficiently small $\sigma^2$, there holds that:
\begin{equation}\label{stable}
F(1,\sigma^2)>0.
\end{equation}
Evaluating $F(1,\sigma^2)$ yields:
$$
F(1,\sigma^2)=\sigma^2 \left((be-af)(ch+dg)-(eg-fh) \right).
$$
Therefore, a sufficient condition for $(\ref{stable})$ to hold, is that:
$$
(fh-eg) - ((af-be)(ch+dg))>0 
$$
After simplifying we found that 
$$
(fh-eg) - ((af-be)(ch+dg)) = f_M'(z^*_M)f'_N(z^*_N)s_Ms_N(1-pq)(x^*_Ns_M+x^*_Ms_N) >0
$$
 This concludes that the dominant eigen value of the Jacobian matrix corresponding to the Darwinan system  \eqref{eq16}-\eqref{eq17} is less than one in modulus.

\end{document}